\documentclass[preprint,12pt,english]{elsarticle}
\usepackage{amsmath}
\usepackage{amssymb}
\usepackage{lipsum}
\usepackage[sans]{dsfont}
\usepackage{multirow}
\usepackage{afterpage}
\usepackage{amsfonts}
\usepackage{bm}%
\usepackage{babel}
\usepackage{hyperref}
\usepackage{graphicx}
\usepackage{color}
\usepackage{amsthm}
\def\fnum@figure{\figurename\thefigure}
\renewcommand{\figurename}{Fig.}

\newcommand{\sech}{\textrm{sech}}
\newtheorem{theorem}{Theorem}

\newtheorem{lemma}[theorem]{Lemma}
\journal{Physics Letter A}

\title{Fidelity based purity and coherence for quantum states}
\begin{document}
\author{Indrajith V S$^*$ , R. Muthuganesan$^\ddag$ R. Sankaranarayanan$^*$ }
\address{$^*$ Department of Physics, National Institute of Technology\\ Tiruchirappalli-620015, Tamil Nadu, India.}
\address{$^\ddag$ Centre for Nonlinear Science and Engineering, School of Electrical and Electronics Engineering,  SASTRA Deemed University, Thanjavur - 613401, Tamil Nadu, India.}
\begin{abstract}
Purity and coherence of a quantum state are recognized as useful resources for various information processing tasks. In this article, we propose a fidelity based valid measure of purity and coherence monotone, and establish a relationship between them.
 This formulation of coherence is extended to quantum correlation relative to measurement. We have also studied the role of weak measurement on purity.
  \end{abstract}
\maketitle
\section{Introduction}
A conventional picture of information is related to entropy, which is a subjective approach, in the sense that something that is known to the sender is unknown to the receiver, where the message is considered as information. Whereas in the objective approach of quantum information, the state itself is considered as information. In this approach, since pure state $ \rho =  \lvert \psi \rangle \langle \psi \rvert $ has zero entropy, purity can be considered to quantify the information represented by the state \cite{purity1,purity2}. In this context, a state which is not maximally mixed and is restricted to the preparation of unitary operations is a resource state of purity, and any deviation from the maximally mixed state quantifies the purity in the state. It turns out that the states that are not maximally mixed can be created using resources in pure states. Some of the earlier works on this approach established a relationship between resource theory of purity and thermodynamics, also with quantum correlations, make any analysis on purity worthy \cite{Thermo_purity}.  We can also connect the purity of a state with the Bruckner-Zeilinger information measure \cite {bruk_zlngr1, bruk_zlngr2}.

On the other hand, the resource theory of coherence has also been studied widely \cite{coherce}. Being a basis dependant quantity and arising due to the superposition principle, coherence plays a key role in quantum metrology \cite{Q_metro1, Q_metro2}, quantum biology \cite{q_bio}, quantum thermodynamics \cite{q_thermo}, etc. A handful of coherence measures have been proposed based on trace distance \cite{trace_coh}, entropy \cite{rel_coh}, Hellinger distance \cite{hel_coh} and skew information \cite{skew_coh1,skew_coh2}. Further, coherence is defined relative to quantum measurements wherein the post-measured state can be considered as an incoherent state. Such coherence relative to measurement is exploited as correlated coherence or quantum correlation.


In this article, we show that purity of a state quantified by logarithm of fidelity with maximally mixed state of dimension $d$ is a good measure. This purity measure is calculated for a general bipartite state, ranging between 0 (for maximally mixed state) and 1 (for pure state), which can give an insight into mixedness in a state. The notion of fidelity is also extended to define coherence monotone by optimizing the closeness between a quantum state and a maximally incoherent state. By establishing a one-to-one relationship between purity and coherence monotone, purity is interpreted as maximal coherence.   

As an extended application, we use this purity measure to classify the quantum state along with coherence. Finally, we discuss the role of weak measurements on the purity measure as well. 
\section{Purity as a resource}
When a state of a quantum system is precisely known, it is an extreme point of a convex set such that it cannot be represented as a convex combination of other states. Such a state is called as a pure state. On the hand, a mixed state permits convex combination and so is generic.
Any state has some purity if it is reluctant to the preparation of completely mixed states. Hence purity of a state can be considered as a resource, and any operation that does not increase purity is a free operation. Here we define some of the free operations as listed below \cite{purity2}.
\begin{itemize}
  \item[$\diamond$] Mixture of unitary operations
  \begin{equation}
    \Xi_U(\rho) = \sum_ip_iU_i\rho U_i^{\dagger}
  \end{equation}
  for all unitary operations $U_i$, with $p_i \geq 0,\sum_ip_i = 1.$
  \item[$\diamond$] Noisy operations
  \begin{equation}
    \Xi_{O}(\rho) = \text{tr}_E[U(\rho\otimes\mathbb{I}/d)U^{\dagger}] 
  \end{equation}
  where $\mathbb{I}$ is the identity matrix.
  \item[$\diamond$] Unital operation - which preserves maximally  mixed state
  \begin{equation}
    \Xi_\mathcal{I}\left(\frac{\mathbb{I}}{d}\right) = \frac{\mathbb{I}}{d}
    \end{equation}
\end{itemize}
 where $d$ is the dimension of the system. These operations preserve the dimension of Hilbert space $\mathcal{H}$, having a hierarchy \cite{purity}
 \begin{equation}
   \{ \Xi_U(\rho)\} \subset \{\Xi_{O}(\rho)\}   \subset   \left\{\Xi_\mathcal{I}\left(\frac{\mathbb{I}}{d}\right)\right\}.
 \end{equation} 
 For any state $\rho$, purity monotone $\mathcal{P}(\rho)$ satisfies the following axioms \cite{purity1,purity2,purity}.
 \begin{itemize}
\item {\textit{Non-negative}:}
$\mathcal{P}(\rho) \geq 0$, with equality holds for maximally mixed state $\mathbb{I}/d$.
 \item{\textit{Monotonicity:}}
 $\mathcal{P}(\rho)$ does not increase under unital operations i.e., $\mathcal{P}(\Xi_{\mathcal{I}}(\rho)) \leq \mathcal{P}(\rho) $, for any unital operation $\Xi_\mathcal{I}$.\\
\item{\textit{Additivity:}}
 For any two states $\rho$ and $\sigma$, the purity of the combined state is 
 \\ $\mathcal{P}(\rho \otimes \sigma) = \mathcal{P}(\rho) + \mathcal{P}(\sigma)$\\
\item{\textit{Normalization:}}
 For any pure state $\rho =\lvert \psi\rangle\langle \psi \rvert $ of $d$ dimension,
 $\mathcal{P}(\rho) = \log_2 d$.
  \end{itemize}
 While the first two axioms are enough to show that $ \mathcal{P}(\rho)$ can be considered as a purity monotone, additivity and normalization make the quantity a purity measure. It should be noted here that the linear purity $ \text{tr}(\rho^2) $ is a purity monotone. Similarly, quantities based on distance measure and entropic measure have also been studied as purity monotones \cite{purity}. Distance measure based purity monotone is defined as
 \begin{equation*}
   \mathcal{P}_\mathcal{D}(\rho) = \mathcal{D}(\rho,\mathbb{I}/d)
 \end{equation*}
where $ \mathcal{D}(\rho,\sigma)$ represents any distance measure between $\rho$ and $\sigma$. Here, the maximally mixed sate is unique such that optimization is not necessary. Considering Hilbert-Schmidt measure as a valid distance measure, distance-based purity monotone for an arbitrary quantum state $\rho$ is written as
\begin{equation}
  \mathcal{P}_H(\rho) = \| \rho - \mathbb{I}/{d} \|^2 \label{h_putity_def}
\end{equation}
where $\|\mathcal{A} \| = \sqrt{\text{tr}(\mathcal{A}^\dagger\mathcal{A})}$ is the Hilbert-Schmidt norm of an operator $ \mathcal{A}$. With this we have 
\begin{align}
    \mathcal{P}_H(\rho) &= \text{tr}\rho^2 - \frac{1}{d}. \label{h_putity} 
  \end{align}
This purity monotone takes the form of Brukner-Zeilinger information measure $\mathfrak{I}(\rho)$, which is a measure of the information content of the system \cite{bruk_zlngr1,bruk_zlngr2}. For a maximally mixed state, the information is zero and for a pure state the information is $ {(d-1)}/ {d}$.
\section{Purity based on fidelity} 
 Fidelity is defined as a measure of closeness between two arbitrary quantum states $\rho$ and $\sigma$. One definition of fidelity is as given below \cite{jozsa}
  \begin{equation}
  \mathcal{F}_u(\rho,\sigma) = \bigg(\text{tr}\sqrt{\sqrt{\sigma}\rho\sqrt{\sigma}}\bigg)^2. \label{ulman}
   \end{equation}
Since the notion of fidelity is taken to many quantum information processing tasks such as teleportation, dense coding, quantum state tomography, quantum phase transition etc., various forms of fidelity have been proposed. Here we use one alternate form of fidelity \cite{fid2}
  \begin{equation}
    \mathcal{F}(\rho,\sigma) = \frac{\big(\text{tr}(\rho~\sigma)\big)^2}{\text{tr}(\rho^2)\text{tr}(\sigma^2)},\label{fidlty}
  \end{equation} 
which is easier to compute than the earlier one. The fidelity $\mathcal{F}(\rho,\sigma)$ defined above has the following properties.
\begin{enumerate}
  \item[$(\mathcal{F}1)$] $0 \leq \mathcal{F}(\rho,\sigma) \leq 1$, maximum is attained when $\rho = \sigma$.
  \item[$(\mathcal{F}2)$] $\mathcal{F}(\rho,\sigma) = \mathcal{F}(\sigma, \rho) $.
  \item [$(\mathcal{F}3)$]Fidelity is invariant under unitary transformation $U$ such that $\mathcal{F}(\rho,\sigma) = \mathcal{F}(U \rho U^{\dagger},U \sigma U^{\dagger})$.
  \item [$(\mathcal{F}4)$]For a pure state defined by $\sigma = | \psi \rangle \langle \psi |  $, fidelity takes the form $\mathcal{F}(\rho, | \psi \rangle \langle \psi |) = \langle \psi | \rho | \psi \rangle / \text{tr}(\rho^2) $.
  \item[$(\mathcal{F}5)$] $\mathcal{F}(\rho_1 \otimes \rho_2, \sigma_1 \otimes \sigma_2) = \mathcal{F}(\rho_1, \sigma_1)\cdot \mathcal{F}(\rho_2, \sigma_2)$.
  \item [$(\mathcal{F}6)$] For any orthogonal projectors \{$\Pi_i\} = \{\lvert i\rangle \langle i\rvert\}$, $\mathcal{F}\big(\sum_i\rho_i,\sum_i\sigma_i\big) = \sum_i \mathcal{F}(\rho_i,\sigma_i)$, where $\rho_i = \Pi_i \rho \Pi_i$ and $\sigma_i = \Pi_i \sigma \Pi_i$.
 \end{enumerate}
 
 Let us consider $ \mathcal{F}(\rho,\mathbb{I}/{d})$ as fidelity between any state $\rho$ amd maximally mixed state $\mathbb{I}/d$. It is clear that ${1}/{d} \leq \mathcal{F}(\rho,\mathbb{I}/{d}) \leq 1$ with minimum and maximum value correspond to pure and maximally mixed states respectively. With this we define a purity measure as
  \begin{equation}
    \mathcal{P}_\mathcal{F}(\rho) = -\log_d \mathcal{F}(\rho,\mathbb{I}/{d}) \label{purity}.
  \end{equation}
 Taking logarithm normalizes the fidelity for being a purity measure such that it is maximum for a pure state and 0 for a maximally mixed state. Hence the measure of purity based on logarithm of fidelity quantifies the mixedness of the system. Since the base of the $\log$ is taken to be the dimension of Hilbert space, this quantity is independent of $d$. Also, this measure can be used for any qudit system. The above purity measure possesses the following properties.
\begin{enumerate}
\item[$(\mathcal{P}1)$] $\mathcal{P}_\mathcal{F}(\rho)$ is a monotonic function, which increases from a maximally mixed state to a pure state.
 \item[$(\mathcal{P}2)$]  $0 \leq \mathcal{P}_\mathcal{F}(\rho) \leq 1$ and is independant of dimension $d$.
\item[$(\mathcal{P}3)$]  $\mathcal{P}_\mathcal{F}(\rho) $ does not increase under unital operation.
\\$\mathcal{P}_\mathcal{F}[\Xi_\mathcal{I}(\rho)] \leq \mathcal{P}_\mathcal{F}(\rho) $, with equality holds for maximally mixed state.
\item[$(\mathcal{P}4)$]  For any two states $\rho$ and $\sigma$ each of dimension $d$, 
$\mathcal{P}_\mathcal{F}(\rho \otimes \sigma) = \mathcal{P}_\mathcal{F}(\rho) +  \mathcal{P}_\mathcal{F}(\sigma) $
\begin{proof}
Taking in account of the property of trace,
\begin{align*}
  \mathcal{P}_\mathcal{F}(\rho \otimes \sigma) =& -\log_{d} \mathcal{F}(\rho \otimes \sigma, \mathbb{I}/d \otimes \mathbb{I}/d)\\
  =& -\log_d \frac{\big(\text{tr}(\rho \otimes \sigma~\frac{\mathbb{I}}{d}\otimes \frac{\mathbb{I}}{d})\big)^2}{\text{tr} (\rho \otimes \sigma)^2~\text{tr}(\frac{\mathbb{I}}{d}\otimes \frac{\mathbb{I}}{d})^2}\\
  =& -\log_{d}\frac{(\text{tr}\rho \frac{\mathbb{I}}{d})^2}{\text{tr} \rho^2 \text{tr} (\frac{\mathbb{I}}{d})^2}\frac{(\text{tr}\sigma \frac{\mathbb{I}}{d})^2}{\text{tr} \sigma^2 \text{tr} (\frac{\mathbb{I}}{d})^2}\\
  =& -\log_{d} \mathcal{F}(\rho,\mathbb{I}/d) - \log_{d} \mathcal{F}(\sigma,\mathbb{I}/d)\\
  =& \,\mathcal{P}_\mathcal{F}(\rho) +   \mathcal{P}_\mathcal{F}(\sigma). 
\end{align*}
\end{proof}
\end{enumerate}
With all these properties, $ \mathcal{P}_\mathcal{F}(\rho)$ can be considered as a purity measure. The measure of purity in eq. (\ref{purity}) can also be written as 
\begin{equation}
  \mathcal{P}_\mathcal{F}(\rho) = \log_d \big(d\,\text{tr}\rho^2\big)\label{log_purity} .
\end{equation}
This relation shows that the proposed fidelity based purity is one-to-one with linear purity $ \text{tr}\rho^2$.
\begin{lemma}
Consider a general bipartite state of the form $\rho = \sum_{i,j} \gamma_{ij} X_i \otimes Y_j$ where \{$X_i : i = 1, 2, \cdots, m^2\}$ and \{$Y_j : j = 1, 2, \cdots, n^2\}$ are set of self-adjoint orthonormal operators on Hilbert spaces $\mathcal{H}^a ~\text{and} ~\mathcal{H}^b$ respectively with $\text{tr}(X_i^{\dagger}X_j) = \text{tr}(Y_i^{\dagger}Y_j)= \delta_{ij}$ and $X_1 = \mathbb{I}/\sqrt{m}$, $Y_1 = \mathbb{I}/\sqrt{n}$. Then the purity is given by 
\begin{equation*}
\mathcal{P}_{\mathcal{F}} (\rho)  = \log_d (d\,\| \Gamma \|^2)
\end{equation*}
where $\Gamma = \gamma_{ij} = \text{tr}(\rho X_i \otimes Y_j)$.
\end{lemma}
\begin{proof}
The fidelity between bipartite mixed state and maximally mixed state is calculated as 
\begin{align*}
  \mathcal{F}(\rho,\mathbb{I}/d) &= \frac{(1/d^2)(\text{tr}\rho)^2}{(1/d)\text{tr}\rho^2}\\
                                 &= \frac{1}{d\,\text{tr}\rho^2}. 
\end{align*}
Since
\begin{equation*}
  \text{tr}\rho^2 = \sum_{ii'jj'} \gamma_{ij}\gamma_{i'j'} X_iX_{i'} \otimes Y_j Y_{j'} = \sum_{ii'jj'} \gamma_{ij}\gamma_{i'j'} = \text{tr}(\Gamma \Gamma^t) = \| \Gamma \|^2
\end{equation*} 
 we have
\begin{equation*}
  \mathcal{P}_{\mathcal{F}}(\rho) = \log_d (d\| \Gamma\|^2).
\end{equation*}
\end{proof}
Having computed fidelity purity for general bipartite state, here we calculate the purity for well-known families of quantum states.
\begin{figure}[!ht]
 \includegraphics[width=0.5\linewidth]{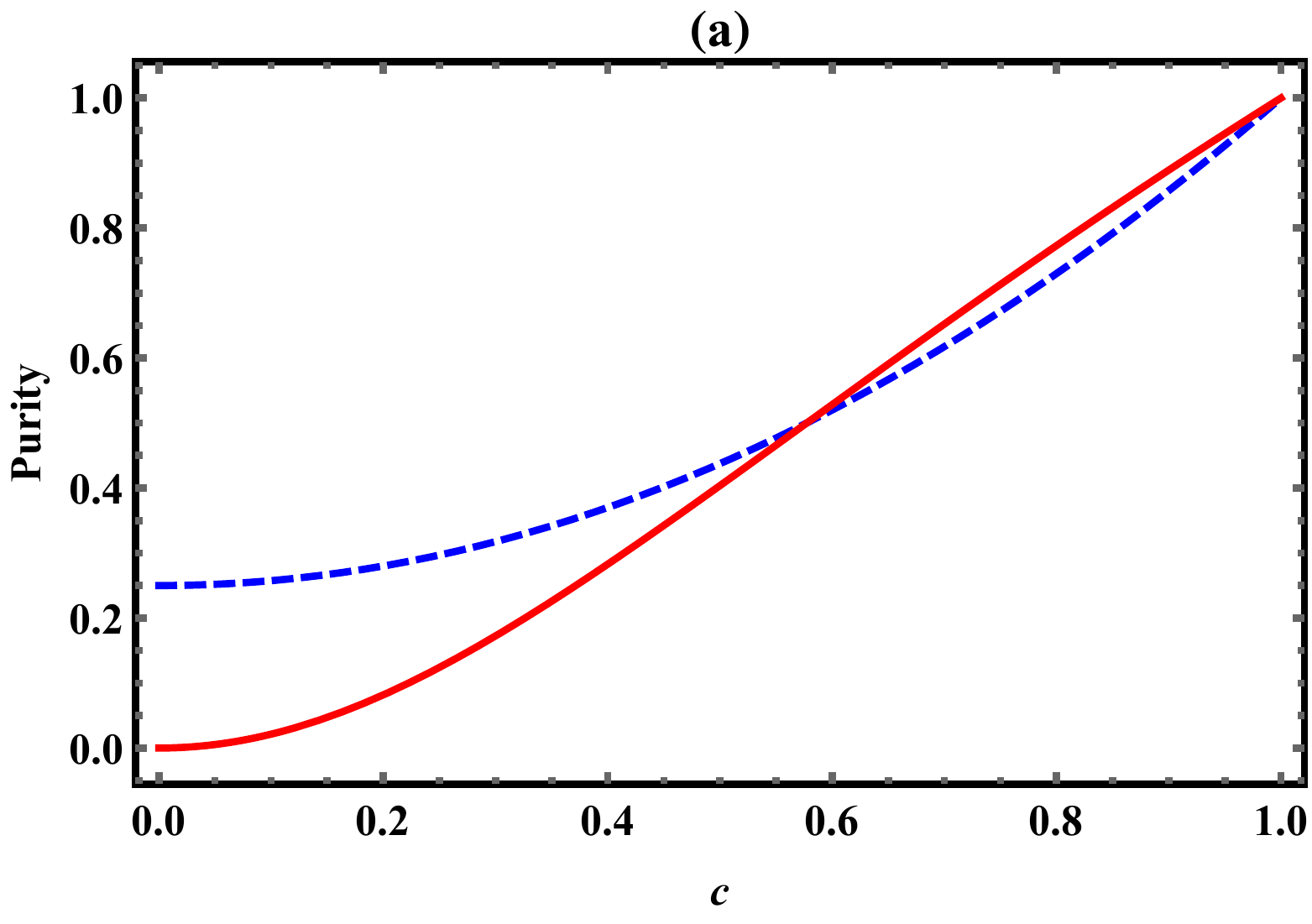}
 \includegraphics[width=0.5\linewidth]{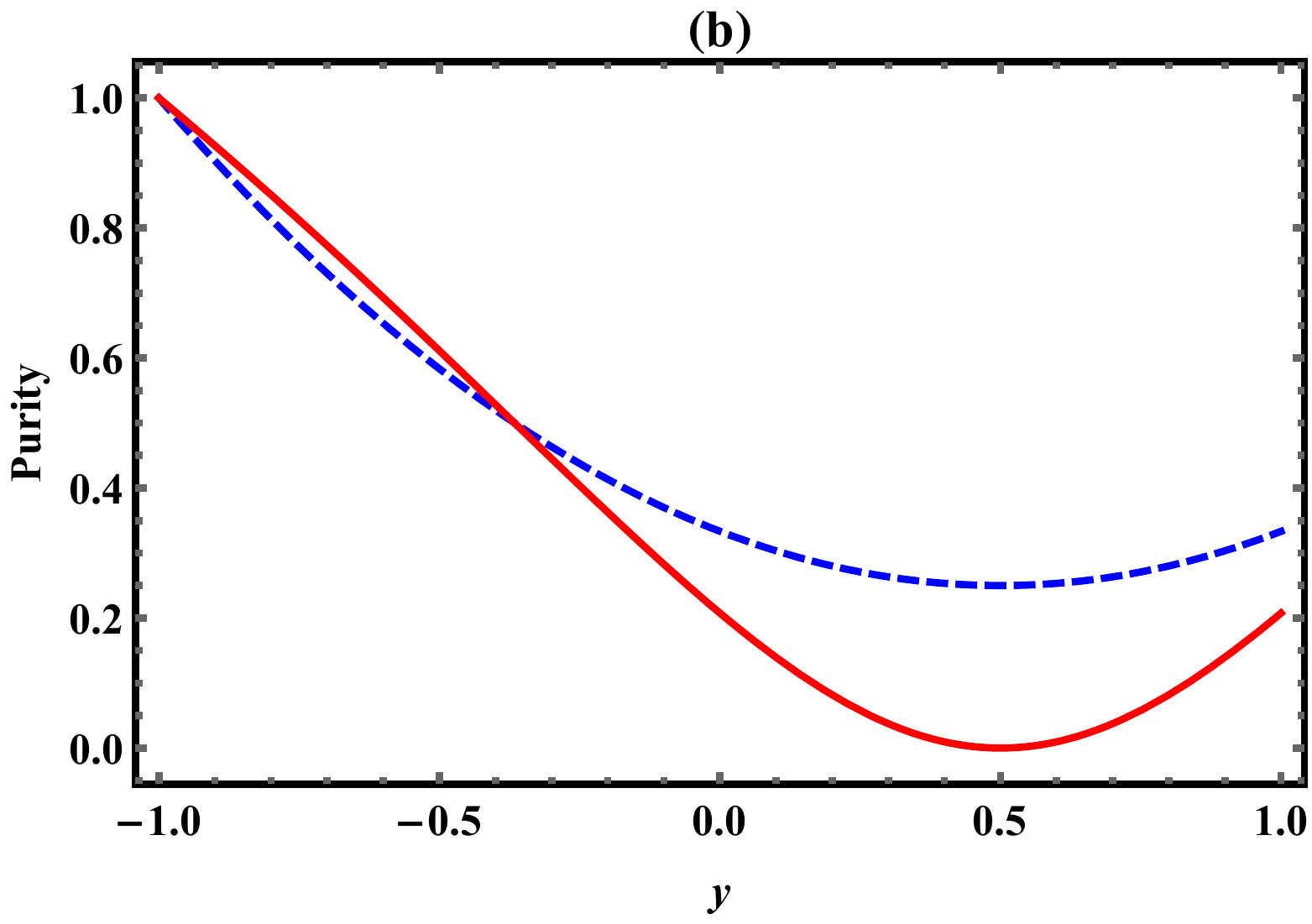}\\
\caption{(color online) Purity quantified  by fidelity  (solid), linear purity (dashed) for (a) Bell states with $ c_1 = c_2 = c_3 = -c$ (b) Werner states.}
\label{pure} 
\end{figure}
\subsection*{Bell states}
The Bloch vector representation of Bell diagonal states is given as 
\begin{equation}
  \rho_{BD} = \frac{1}{4}\Big(\mathbb{I}\otimes \mathbb{I} + \sum^{3}_{i = 1} c_i(\sigma_i \otimes \sigma_i)\Big) \label{bell_diag}
\end{equation}
where $\textbf{c} = (c_1, c_2, c_3) $ is the correlation coefficients with $ -1 \leq c_i \leq 1$. For these states the linear purity is given 
\begin{equation}
   \text{tr}(\rho_{BD}^2) =\frac{1}{4}(1+c_1^2 + c_2^2 + c_3^2).
\end{equation}
\subsection*{Werner States}
Werner states with $d\times d$ dimension can be represented as,
\begin{equation}
  \rho_{W} = \frac{d-y}{d^3-d}\mathbb{I} + \frac{yd-1}{d^3-d}\sum_{\alpha \beta}\lvert \alpha \rangle \langle\beta \rvert  \otimes\lvert \beta \rangle \langle \alpha  \rvert
\end{equation}
where $\sum_{\alpha \beta}\lvert \alpha \rangle \langle\beta \rvert  \otimes\lvert \beta \rangle \langle \alpha  \rvert$ is flip operator with $y \in [-1,1]$. Linear purity for the above state is calculated as
\begin{equation}
  \text{tr}(\rho_W^2)= \frac{1}{9}(4y^2 - 2y + 1).
\end{equation}

In Fig.\ref{pure}, we plot purity quantified by fidelity for Bell states and Werner states. We observe that the proposed purity is more sensitive with state parameter. Unlike linear purity which does not vanish, fidelity purity is zero for a maximally mixed state.
 \subsection{Classification of quantum states}
  An arbitrary $d$ dimensional quantum  state can be represented as 
  \begin{equation}
    \rho = \frac{\mathbb{I}}{d} + \frac{1}{2}\sum^{d^2-1}_{i=1} x_i X_i
  \end{equation}
where $x_i = \text{tr}(\rho\,X_i)$ with $X_i$ being the generator of $SU(d)$. These generators are hermitian and  satisfiy the conditions $\text{tr}(X_i) = 0 $, $ \text{tr}(X_iX_j) = 2 \,\delta_{ij}$. Then the purity of the above state is given by 
  \begin{align}\nonumber
    \mathcal{P}_\mathcal{F}(\rho) &= -\log_d (\text{tr}\rho^2)\\ 
                                  &= \log_d \left(1+\frac{d}{2}\sum^{d^2-1}_ix^2_i\right).
  \end{align}
 One can define $l_1$ coherence as  $C_{l_1}(\rho) = \sum_{ i \neq j}\lvert \rho_{ij} \rvert$, such that $ 0 \leq C_{l_1} \leq (d-1)$ \cite{coherce,coh2}. This quantity is computed for the above state as
  \begin{equation}
    C_{l_1}(\rho) = \sum^{k}_{ i = 1}\sqrt{x^2_i + x^2_{i+k}}
  \end{equation}
  where $ k = (d^2 -d)/2$. While purity gives the inference about the diagonal elements, coherence gives information about the off-diagonal elements. In order to classify a given state $\rho$, it is useful to define
  \begin{equation*}
   \tau(\rho) = \frac{1}{2}\left(\mathcal{P}_{\mathcal{F}}(\rho) + \frac{1}{d-1}   C_{l_1}(\rho)\right)
  \end{equation*} 
  such that $ 0 \leq \tau(\rho) \leq 1$. Here $ \tau(\rho) = 0$ 
corresponds to maximally mixed (impure) incoherent states and $ \tau(\rho) = 1$ refers to pure coherent state.
\section{Coherence based on Fidelity}
Coherence is an important quantity in quantum information processing tasks. A state is considered as incoherent (with respect to a basis), if it is diagonal and is denoted as $ \delta = \sum^d_i \delta_i \lvert i \rangle \langle i \rvert $ with $\{\lvert i \rangle\}$  being a set of orthonormal basis for dimension $d$. Using the fidelity defined in  eq.(\ref{fidlty}), here we define a coherence monotone as
\begin{eqnarray}
  C_{\mathcal{F}}(\rho) =& ^{\text{min}}_{\delta \in \mathcal{I}}\big(1-\mathcal{F}(\rho,\delta)\big)\nonumber \\
                        =& 1- \,^{\text{max}}_{\delta \in \mathcal{I}}\mathcal{F}(\rho,\delta)\label{cohrce_fdlty}
\end{eqnarray}
where $\mathcal{I} = \{\delta\}$ is a set of incoherent states. For being a bonafied coherence monotone, $ \mathcal{C}_\mathcal{F}(\rho)$ must satisfy the following properties.
\begin{enumerate}

  \item[$(\mathcal{C}1)$] $C_{\mathcal{F}}(\rho) \geq 0$, with equality holds for $\rho \in \mathcal{I}$ 
  \begin{proof}
  From the properties of the fidelity defined above, $0 \leq \mathcal{F}(\rho,\delta) \leq 1$. $\mathcal{F}(\rho,\delta) = 1$ if and only if $\rho = \delta \in \mathcal{I}$ 
  \end{proof}
  \item[$(\mathcal{C}2)$] \textit{Weak monotonicity}: Monotoncity under incoherent complete positive trace preserving (ICPTP) map $ \Phi(\rho)$. That is $C_{\mathcal{F}}(\rho) \geq C_{\mathcal{F}}[\Phi(\rho)]$.
  \begin{proof}
  An incoherent complete positive trace preserving operation is defined as $\Phi(\rho) = \sum_k A_k \rho A^{\dagger}_k$ with Kraus operators $ A_k$ such that \\$\sum_k A^{\dagger}_k A_k = \mathbb{I}$, where $A_k \mathcal{I}A^{\dagger}_k \subseteq \mathcal{I}$. The CPTP map is expressed in terms of unitary operator $U$ as 
 \begin{equation}
   U(\rho  \otimes| k \rangle \langle k | ) U^{\dagger}
 \end{equation} 
where the environment is written in the basis $ |k \rangle $.
 \begin{align*}
   \mathcal{F}[\Phi(\rho),\Phi(\delta)] =&\, \mathcal{F}[\text{tr}_E\big(U(\rho\otimes| k\rangle \langle k | )U^{\dagger}\big), \text{tr}_E\big(U(\delta \otimes| k\rangle  \langle k | )U^{\dagger}\big) ] \\
      \geq&\, \mathcal{F}(\rho \otimes| k\rangle \langle k|,\delta \otimes | k\rangle \langle k| ).
  \end{align*}
  While writing the last inequlity, we employ the fact that $\mathcal{F}(\rho_k, \sigma_k) \geq \mathcal{F}(\rho, \sigma)$ \cite{marg_state}, where $\rho_k$ and $ \sigma_k$ are the marginal states of the states $\rho$ and $\sigma$ respectively. Then by the property of fidelity ($ \mathcal{F}5$)
  \begin{equation}
    \mathcal{F}[\Phi(\rho),\Phi(\delta)] \geq \mathcal{F}(\rho,\delta)
  \end{equation}
  we complete the proof.
  \end{proof}
\item[$(\mathcal{C}3)$] \textit{Strong monotonicity}: $C_\mathcal{F}(\rho) \geq \sum_k p_k C_\mathcal{F}(\rho_k)$, where $\rho_k = A_k \rho A^{\dagger}_k/ p_k$ \\ and $p_k = \text{tr} (A_k \rho A^{\dagger}_k)$.
 \begin{proof}
Stinespring dialation theorem \cite{Stinespring,Stinespring2}, is stated as follows.
   
 Let $T:\mathcal{S}(\mathcal{H})\rightarrow \mathcal{S}(\mathcal{H})$ \textit{be a CPTP map on finite dimensional Hilbert space $\mathcal{H}$. Then there exist another Hilbert space $\mathcal{K}$ and unitary operation $U$ on $\mathcal{H}\otimes \mathcal{K}$ such that }
   \begin{equation}
     T(\rho) = \text{tr}_{\mathcal{K}}[U (\rho \otimes \lvert 0 \rangle \langle 0 \rvert )U^{\dagger}]
   \end{equation}
   \textit{$\forall~\rho \in \mathcal{S}(\mathcal{H})$. Here $\mathcal{K}$ is chosen such that $ dim\, \mathcal{K} \leq dim \,2\mathcal{H} $ }. Let $ \Phi(\rho) = \sum_k A_k \rho A^\dagger_k$ be a CPTP map and we choose an ancillary state $ \lvert 0 \rangle \langle 0 \rvert \in \mathcal{K}$, the extended Hilbert space.
   Stinespring representation can be used for Kraus decompostion as
   \begin{equation}
     A_k \rho A^{\dagger}_k = \text{tr}_\mathcal{K} \bigg[\mathbb{I}\otimes\lvert k \rangle \langle k \rvert U(\rho \otimes\lvert 0 \rangle \langle 0 \rvert)U^{\dagger}\mathbb{I}\otimes\lvert k \rangle \langle k \rvert\bigg].
   \end{equation}
   Taking this representation to fidelity as\\
   $
     \sum_k \mathcal{F}\bigg(A_k\rho A^{\dagger}_k,A_k\sigma A^{\dagger}_k\bigg) = \sum_k \mathcal{F}\bigg(\text{tr}_\mathcal{K}\big[ \mathbb{I}\otimes\lvert k \rangle \langle k \rvert U(\rho \otimes\lvert 0 \rangle \langle 0 \rvert)U^{\dagger}\mathbb{I}\otimes\lvert k \rangle \langle k \rvert \big],\\ \hspace*{5cm}\text{tr}_\mathcal{K}\big[\mathbb{I}\otimes\lvert k \rangle \langle k \rvert U(\sigma \otimes\lvert 0 \rangle \langle 0 \rvert)U^{\dagger}\mathbb{I}\otimes\lvert k \rangle \langle k \rvert \big]\bigg).
   $
   
   Since fidelity obeys monotonicity under CPTP maps, the above equation is written as \\
   $
   \sum_k \mathcal{F}\bigg(A_k\rho A^{\dagger}_k,A_k\sigma A^{\dagger}_k\bigg) \geq \sum_k \mathcal{F}\bigg( \mathbb{I}\otimes\lvert k \rangle \langle k \rvert U(\rho \otimes\lvert 0 \rangle \langle 0 \rvert)U^{\dagger}\mathbb{I}\otimes\lvert k \rangle \langle k \rvert ,\\ \hspace*{5cm}\mathbb{I}\otimes\lvert k \rangle \langle k \rvert U(\sigma \otimes\lvert 0 \rangle \langle 0 \rvert)U^{\dagger}\mathbb{I}\otimes\lvert k \rangle \langle k \rvert \bigg). \\
   $ Making use of the property $ (\mathcal{F}6)$, we get
   \begin{align*}
  \sum_k \mathcal{F}\bigg(A_k\rho A^{\dagger}_k,A_k\sigma A^{\dagger}_k\bigg) & \geq \mathcal{F}(U \rho\,\otimes\rvert 0\rangle \langle 0\rvert U^{\dagger},U \sigma \,\otimes\rvert 0\rangle \langle 0\rvert U^{\dagger}) \\
   &\geq \mathcal{F}(\rho,\sigma),
   \end{align*}
  \\ implying that $C_\mathcal{F}(\rho) \geq \sum_k p_k C(\rho_k)$.
\end{proof}
  \end{enumerate} 
 It should be noted that unlike Ulmann fidelity, 
 the fidelity in eq.(\ref{fidlty}) is not a concave function \cite{fid2}. Hence coherence defined in terms of fidelity does not obey the convexity making it only a coherence monotone, and not a bonafide coherence measure \cite{coh_monotone}.
 \begin{theorem}
    Maximal coherence $ C_m(\rho) = \,^{\text{sup}}_U C_\mathcal{F}(U\rho U^\dagger)$ is the upper bound for the coherence based on fidelity and is given by
    \begin{equation}
      C_m(\rho) = 1- d^{-\mathcal{P}_\mathcal{F}} \geq C_\mathcal{F}(\rho).
    \end{equation}
  \end{theorem}
  \begin{proof}
  Here we have $\mathbb{I}/d = \delta \in \mathcal{I}$, where maximally mixed state can be considered as a maximally incoherent state. This makes $\delta$ as unitary invariant such that $ \delta = U \delta U^\dagger$. Hence the maximal coherence can be written as \cite{purity}
  \begin{equation}
   C_m(\rho) =  \,^{\text{sup}}_U \left(1-\mathcal{F}(U \rho U^\dagger,U (\mathbb{I}/d) U^\dagger)\right).
  \end{equation}
By the property of fidelity, which is unitary invariant, we have
\begin{equation}
   C_m(\rho) =  1-\mathcal{F}(\rho ,\mathbb{I}/{d} ) =  1-d^{-\mathcal{P}_\mathcal{F}}.
  \end{equation} 
    \end{proof} 
  Alternalively
  \begin{equation*}
    \mathcal{P}_\mathcal{F}(\rho) = \log_d(1- C_m(\rho))^{-1}
  \end{equation*}
which implies that the purity measure has one to one relation with maximal coherence, i.e. the upper bound of coherence.
In what follows, we analyze the role of von Neumann measurement on purity.
For a state $\rho$, the post-measurement state on subsystem can be written as $\Pi(\rho) = \sum_k(\Pi_k\otimes \mathbb{I})\rho (\Pi_k \otimes \mathbb{I})$, where $  \{\Pi_k = \lvert k \rangle \langle k \rangle\}$ is a set von Neumann projective measurements on subsystem $a$. If a state $\rho$ is a coherent pure state, then the projective measurement on the state makes the resultant state an incoherent mixed one. In other words, for any pure state, we have
\begin{equation}
  \mathcal{P}_\mathcal{F}(\rho) \geq   \mathcal{P}_\mathcal{F}\left(\Pi(\rho)\right)
\end{equation}
with equality holds when $\rho$ is an incoherent state. This implies that von Neuman measurement reduces the purity with equality holds only if $\rho$ is an incoherent state. Here $ \Pi(\rho)$ need not be a maximally mixed state, rather it is a diagonal state. Since
\begin{equation}
  \mathcal{F}(\rho,\mathbb{I}/d) \leq \mathcal{F}\left(\rho,\Pi(\rho)\right)
\end{equation}
we have
\begin{align*}
  1- \mathcal{F}\left(\rho,\mathbb{I}/{d}\right) &\geq 1- \,^{\text{min}}_{\Pi^a}\mathcal{F}(\rho,\Pi(\rho))
  \end{align*}
  and hence
  \begin{align*}
  C_m(\rho)                   &\geq  N_\mathcal{F}(\rho)
\end{align*}
where $  N_\mathcal{F}(\rho)$ is the measurement induced nonlocality based on fidelity (F-MIN) \cite{fmin}.
Alternalively we have
\begin{equation}
\mathcal{P}_\mathcal{F}(\rho) \geq \log_d \left(1- N_\mathcal{F}(\rho) \right)^{-1}.
\end{equation}
\section{Quantum correlation and purity}
  Consider a bipartite state $ \rho$ shared by parties $a$ and $b$ in $ \mathcal{H}^a$ and $ \mathcal{H}^b$ respectively. Coherence based on fidelity relative to measurement is given by
  \begin{equation}
    C_{\mathcal{F}}(\rho|\Pi) = 1- \mathcal{F}(\rho,\Pi(\rho))
  \end{equation}
Using this, we define the difference between the coherence of global state $\rho$ and coherence of the product states as
  \begin{equation}
    \Delta_\mathcal{F}(\rho|\Pi) = C_\mathcal{F}(\rho|\Pi) - C_\mathcal{F}(\rho^a \otimes \rho^b|\Pi)
  \end{equation}
 Using the properties of fidelity, the coherence of product state can be written as 
 \begin{align*}
   C_\mathcal{F}(\rho^a\otimes\rho^b|\Pi) &= 1- \mathcal{F}(\rho^a\otimes\rho^b,\Pi(\rho^a\otimes\rho^b)) \\
                                          &= 1- \mathcal{F}\left(\rho^a\otimes\rho^b, (\Pi^a\otimes\mathbb{I})(\rho^a\otimes\rho^b)(\Pi^a\otimes\mathbb{I})\right) \\
                                          &= 1- \mathcal{F}(\rho^a,\Pi^a(\rho^a))\cdot\mathcal{F}(\rho^b,\rho^b)\\
                                          &= C_\mathcal{F}(\rho^a|\Pi^a).
 \end{align*} 
Quantum correlation based on the difference in coherence is defined as
  \begin{equation}
    \mathcal{Q}_\mathcal{F}(\rho) =\, ^{\text{min}}_\Pi\,\Delta_\mathcal{F}(\rho|\Pi)
  \end{equation}
  where the minimization is taken over all von Neumann projective measurements. This quantity has the following properties.
 \begin{enumerate}
   \item [($\mathcal{Q}1$)] $\mathcal{Q}_\mathcal{F}(\rho) \geq 0$, for any bipartite state $ \rho$, with equality holds for the product state of the form $ \rho = \sum_k p_k \lvert k \rangle \langle k \rvert \otimes \rho_k$.
   \begin{proof}
   The bipartite state of the form $\rho = \sum_k p_k \lvert k \rangle \langle k \rvert \otimes \rho_k$ is invariant under von Neumann projective measurement such that $ \rho = \Pi(\rho)$. Also the subsystem $ \rho^a = \sum_k p_k \lvert k \rangle \langle k \rvert$ is an incoherent state such that $ C_\mathcal{F}(\rho|\Pi) =C_\mathcal{F}(\rho^a|\Pi^a) = 0 $. By the property of fidelity, the positivity condition holds for $ \mathcal{Q}_\mathcal{F}(\rho)$.
   \end{proof}
   \item [($\mathcal{Q}2$)]$\mathcal{Q}_\mathcal{F}(\rho)$ in invariant under local unitary transformation such that  for any unitary operator $ U $, $ \mathcal{Q}_\mathcal{F}(\rho) = \mathcal{Q}_\mathcal{F}(U \rho U^\dagger)$.
   
   This can be proved by realizing that fidelity in invariant under unitary operations.
   \item[($\mathcal{Q}3$)] Quantum correlation is a non-increasing function under any quantum operation $ \varepsilon(\rho)$ on the state such that  $ \mathcal{Q}_\mathcal{F}(\varepsilon(\rho)) \leq \mathcal{Q}_\mathcal{F}(\rho)$.
   \begin{proof}
   \begin{align*}
     \mathcal{Q}_\mathcal{F}(\varepsilon(\rho)) &= C_\mathcal{F}[\varepsilon(\rho)|\Pi]- C_\mathcal{F}[\varepsilon(\rho^a)|\Pi^a].
   \end{align*}
   Since the coherence is a function of fidelity, it is a contractive function such that $ \mathcal{F}\left(\varepsilon(\rho),\varepsilon(\sigma)\right) \leq \mathcal{F}(\rho,\sigma)$ and hence $C_\mathcal{F}\left(\varepsilon(\rho)|\Pi\right) \leq C_\mathcal{F}(\rho|\Pi)$. Implying
   \begin{equation*}
    \mathcal{Q}_\mathcal{F}(\varepsilon(\rho)) \leq \mathcal{Q}_\mathcal{F}(\rho).
   \end{equation*}
   \end{proof}
 \end{enumerate}
 If $ \Pi(\rho) = \mathbb{I}/d$, the maximally mixed state, then
 \begin{equation}
    \Delta_\mathcal{F}(\rho|\mathbb{I}/d) = d_a^{-\mathcal{P}_\mathcal{F}(\rho^a)} - d^{-\mathcal{P} _\mathcal{F}(\rho)} \geq 0
 \end{equation}
 where $ d_a$ is the dimension of the reduced state $\rho^a$.
\section{Weak Measurement}
In what follows, we analyze the role of weak measurement on purity. A quantum measurement with any number of outcome can be performed as sequence of measurement with two outcomes. It can be shown that von Neuman projective measurement can be implemented as sequence of weak measurement. Such weak measurement operators are defined as \cite{weak2} 
\begin{equation}
  \Omega_x = t_1 \Pi^1 + t_2 \Pi^2,~~ \Omega_{-x} = t_1 \Pi^2 + t_2 \Pi^1
\end{equation}
where $ t_{1,2} = \sqrt{\frac{1\pm \tanh x}{2}}$ with $ x \in \mathbb{R}$ being the strength of measurement. The orthonormal projectots $ \Pi^1$ and $ \Pi^2$ satisfy the condition $ \Pi^1 + \Pi^2 = \mathbb{I}$ and they can be decompossed as $ \Pi^1 = \sum^k_{i=1} \Pi_i$, $ \Pi^2 = \sum^n_{i=k+1} \Pi_i$, where $\{\Pi_i\} $ is a set of von Neumann measurements \cite{weak_corltn}. Weak operators reduces to projective measurements in the limit $ x \rightarrow \infty$ and satisfy the relation $ \sum_{j =\pm x} \Omega^\dagger_j \Omega_{j} = \mathbb{I}$.
Considering a state $\rho$, the state after weak the measurement is \cite{hmin}
\begin{align}\nonumber
  \Omega(\rho) &= \sum_{j=\pm x} (\Omega_j \otimes \mathbb{I})\rho(\Omega_j \otimes \mathbb{I})\\
               &= t \rho + (1-t) \Pi(\rho) \label{weak1}
\end{align}
where $t = 2 t_1t_2 = \sech\,x$. From the above equation  it is clear that tr $[\Omega(\rho)] = 1$. This relation is useful to interpret $ \Omega(\rho)$ as convex comination of a state and its post measured state with weights of $t$ and $(1-t)$ respectively. 
Here, we intend to find the closeness of a quantum state with its weakly measured state $ \Omega(\rho)$, using fidelity as
\begin{equation}
  \mathcal{F}\big(\rho,\Omega(\rho)\big) = \frac{\big(\text{tr}(\rho\,\Omega(\rho))\big)^2}{\text{tr}(\rho^2)\text{tr}\big(\Omega^2(\rho)\big)}.
\end{equation}
Combining eq.(\ref{weak1}), along with the identity $ \text{tr}[\Pi(\rho)^2] = \text{tr}[\rho \Pi(\rho)]$, we have \cite{hmin},
\begin{equation*}
  \text{tr}[\Omega^2(\rho)] = \text{tr}[t^2\rho^2 + (1-t^2)\rho\Pi(\rho)].
\end{equation*}
The above equation can also be rewritten as 
\begin{equation*}
  \text{tr}[\Omega^2(\rho)] = \text{tr}[\rho(t\,\Omega(\rho)+(1-t)\Pi(\rho))].
\end{equation*}
Making use of these results, the fidelity can be written as 
\begin{align} \nonumber
  \mathcal{F}(\rho,\Omega(\rho)) &= \frac{\left(\text{tr}[\rho\,\Omega(\rho)]\right)^2}{\text{tr}[\rho^2]\cdot\text{tr}[t\,\rho\,\Omega(\rho) + (1-t)\rho\Pi(\rho)]} \\
                                &= \frac{\left(\text{tr}[\rho\,\Omega(\rho)]\right)}{\text{tr}[\rho^2]\cdot \left(t + (1-t) \zeta \right)}\label{weak_fid}
\end{align}
where $ \zeta =\text{tr}[\rho\,\Pi(\rho)]/{\text{tr}[\rho\,\Omega(\rho)]}$. The fidelity defined above is function of strength of measurement such that 
\begin{equation*}
 \lim_{x \rightarrow \infty} \mathcal{F}(\rho,\Omega(\rho)) = \frac{\text{tr}[\rho\,\Pi(\rho)]}{\text{tr}[\rho^2]},
\end{equation*}
which is the minimum value of the fidelity. Further, $\lim_{x \rightarrow 0} \mathcal{F}(\rho,\Omega(\rho)) = 1$ and thus we have
\begin{equation*}
 \frac{\text{tr}[\rho\,\Pi(\rho)]}{\text{tr}[\rho^2]} \leq \mathcal{F}(\rho,\Omega(\rho)) \leq 1.
\end{equation*}
The purity of the state $ \Omega(\rho)$ is 
\begin{align} \nonumber
 \mathcal{P}_\mathcal{F}(\Omega(\rho)) &=-\log_d \mathcal{F}\left(\Omega(\rho), \mathbb{I}/{d}\right)\\ \nonumber
                                      &= \log_d \left( d\,\text{tr}[\Omega^2(\rho)]\right)\\
                                      &= \log_d \left(d\,\text{tr}[t^2\rho^2 + (1-t^2)\rho\Pi(\rho)]\right).
\end{align}
With this we have the following limiting case:
\begin{equation}
\lim_{x \rightarrow \mu } \mathcal{P}_\mathcal{F}(\Omega(\rho)) = \begin{cases}
\mathcal{P}_\mathcal{F}(\rho)   & \mu = 0 \\
\log_d \left(d\,\text{tr}[\rho\Pi(\rho)]\right) &  \mu = \infty.
\end{cases}
\end{equation}
 \section{Conclusions}  
 In this paper, we propose a fidelity based purity measure of an arbitrary quantum state ranging between 0 to 1. It is also one-to-one with linear purity. In addition, we propose a valid fidelity based coherence monotone, whose maximum is shown to be one-to-one with fidelity-purity. This establishes the connection between purity and coherence of quantum state. Finally, the role of weak measurement on fidelity-purity is also discussed.
\bibliographystyle{99}  

\end{document}